\documentclass[10pt,english]{IEEEtran}%
\usepackage[T1]{fontenc}
\usepackage[fleqn]{amsmath}
\usepackage{graphicx}
\usepackage{amssymb}
\usepackage{amsfonts}
\usepackage{amscd}
\usepackage{multirow}
\usepackage{color}
\usepackage{babel}
\usepackage{subfigure}
\usepackage{algorithm}
\usepackage{algorithmicx}
\usepackage{listings}
\usepackage{epsfig}
\usepackage{algpseudocode}
\newtheorem{theorem}{Theorem}
\newtheorem{lemma}[theorem]{Lemma}

\providecommand{\U}[1]{\protect\rule{.1in}{.1in}}
\makeatletter
\@ifundefined{definecolor}
{
}{}

\newtheorem{thm}{Theorem}[section]
\newtheorem{cor}[thm]{Corollary}

\newtheorem{rem}{Remark}[section]

\makeatother
\begin{document}

\title{Optimal Transmission Policies for Energy Harvesting Two-hop Networks}
\author{\IEEEauthorblockN{Oner Orhan and Elza Erkip}\\
\IEEEauthorblockA{Dept. of ECE, Polytechnic Institute of New York University\\
oorhan01@students.poly.edu, elza@poly.edu}}
\maketitle
\vspace{-0.1in}
\begin{abstract}
In this paper, a two-hop communication system with energy harvesting nodes is considered. Unlike battery powered wireless nodes, both the source and the relay are able to harvest energy from environment during communication, therefore, both data and energy causality over the two hops need to be considered. Assuming both nodes know the harvested energies in advance, properties of optimal transmission policies to maximize the delivered data by a given deadline are identified. Using these properties, optimal power allocation and transmission schedule for the case in which both nodes harvest two energy packets is developed.
\end{abstract}
\vspace{-0.1in}
\section{Introduction}
Recent advances in energy harvesting devices allow wireless sensor networks to operate in a self-powered fashion for extended periods of time. However, due to size of sensors and low ambient energy sources, harvested energy from environment is low \cite{winston}. Hence, efficient utilization of available energy is essential to increase lifetime of sensor networks as well as to obtain maximum amount of information from sensors. Optimal transmission polices for energy harvesting nodes have recently attracted significant interest \cite{Sharma}-\cite{maria}. Sharma et al.\ \cite{Sharma} and Castiglione et al.\ \cite{elza} investigate optimal policies for stochastic energy arrivals. While \cite{Sharma} gives optimal policies under data queue constraints, \cite{elza} generalizes the framework to study power allocation for both source acquisition and transmission under distortion constraints. Other line of work, such as \cite{Yang2010}-\cite{deniz2} studies off-line transmission policies, where the harvested energies are known in a non-causal fashion at the transmitters. The scenarios investigated include single link \cite{Yang2010}, single link with varying channel SNR \cite{chin}, single link fading channel \cite{fade}, multiple access channel \cite{multi}, broadcast channel \cite{broad}, orthogonal half-duplex relay channel \cite{Huang}, and two-hop \cite{deniz} networks. Similarly, single energy harvesting link under battery imperfections is investigated in \cite{Yener} and \cite{deniz2}. The non-causal knowledge of the energy harvesting process may be applicable for predictable energy models (such as \cite{maria}), or more generally allows one to obtain upper bounds of performance.

In this paper, we focus on a two-hop energy harvesting wireless network with a half-duplex relay. Since one of the main sources of energy consumption of a node is the power amplifier \cite{ECWN}, we assume that harvested energies are only used for transmission purposes. We also assume that both the source and the relay know the harvested energies in advance, and that energies arrive with arbitrary amounts at arbitrary times. Moreover, both the source and the relay have infinite size battery and data buffer. Our goal is to find an optimal transmission policy (which includes scheduling and power allocation), subject to energy causality constraints at the source and the relay and data causality constraint at the relay, that maximizes total transmitted data from the source to the destination until a given deadline \emph{T}. Note that \cite{deniz} also studies two-hop communication with a half duplex relay; however only single energy arrival at the source is considered, significantly simplifying the problem. Our framework is more general, in that, we identify properties of the optimal transmission policy for the multi-energy arrival case both at the source and the relay. Then, using these properties, we provide the optimal policy for two energy arrival case. Optimal transmission policy for the multi-arrival case builds upon the two-arrival solution will be discussed in future work.

In the next section, we describe the general system model. In Section \ref{Properties of Optimal Scheduling and Power Allocation}, we describe some properties of the optimal transmission policy. In Section IV, we provide the solution for two energy arrivals to the source and the relay. In Section V, simulation results are presented. We conclude in Section VI.
\vspace{-0.1in}
\section{System Model}
\label{system_model}
We consider a two-hop communication system with energy harvesting source and relay as shown in Figure \ref{fig 1}. The relay is half-duplex. Communication takes place until the deadline T. The source and the relay harvest energies with amounts $E_{s,i}>0$ and $E_{r,j}>0$ at times $t_{s,0}<...<t_{s,M}$ and $t_{r,0}<...<t_{r,N}$, respectively. We set $t_{s,0}=t_{r,0}=0$ and $t_{s,M}=t_{r,N}=T$. Inter-arrival times of energy packets are denoted as $\tau_{s,i}=t_{s,i}-t_{s,{i-1}}$ and $\tau_{r,j}=t_{r,j}-t_{r,{j-1}}$ for all $t_{s,i}, t_{r,j}\in[0,T]$ for the source and the relay, respectively. We assume that the energy arrival profiles are known to all nodes in the network prior to transmission. In addition, it is assumed that each node has an infinite size, ideal battery (instantly chargeable and no leakage) and they consume harvested energies only for transmission purposes. Furthermore, relay has unlimited data buffer.
\begin{figure}[!t]
\center
\includegraphics[scale=1,trim= 0 10 15 0]{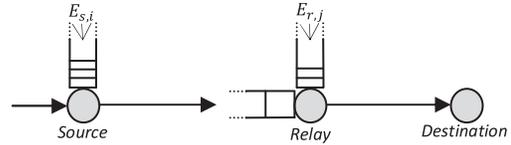} 
\caption{Two-hop system with energy harvesting source and relay}%
\label{fig 1}%
\vspace{-0.2in}
\end{figure}
We consider a non-negative, strictly concave and monotonically increasing power-rate function for each terminal, \emph{R=g(p)}, where $p$ is the instantaneous power \cite{antepli}. These properties are satisfied in many widely used power-rate models, for example, Shannon capacity $g(p)=\log(1+\lvert h\rvert^2p)$ for complex additive white Gaussian noise channel with channel gain $h$ and unit noise variance. We assume that the source and the relay are allowed to change their data rate instantaneously by changing their power, \emph{p(t)}. The source-relay and the relay-destination channel gains $h_s$ and $h_r$ are in general different. However, since the amplitude of the channel gain can be normalized by scaling the harvested energy, without loss of generality, we assume that each link has unity gain and therefore characterized by the same power-rate function.

The transmission {\em schedule} indicates which node is transmitting at a given time. Recall that half-duplex constraint suggests the source and the relay cannot transmit simultaneously. {\em Epoch} $i$ of the source denotes the transmission period of the source when it is continuously scheduled for the $i$'th time. Duration of the epoch $i$ is denoted by $\xi_{s,i}$. Similarly, $\xi_{r,i}$ denotes the duration of $i$'th epoch of the relay. The power allocation functions of the source and the relay depend on harvested energies and are represented by $p_s(t)$ and $p_r(t)$, respectively. Due to half-duplex constraint $p_s(t)p_r(t)=0$ and the transmission schedule can be inferred from the power allocation functions $p_s(t)$ and $p_r(t)$. The total transmitted data by the source and the relay up to time $t$ can be expressed as $B_s(t)=\int_{0}^{t}{g(p_s(\tau))d\tau}$ and $B_r(t)=\int_{0}^{t}{g(p_r(\tau))d\tau}$, while the consumed energies are $E_s(t)=\int_{0}^{t}{p_s(\tau)d\tau}$ and $E_r(t)=\int_{0}^{t}{p_r(\tau)d\tau}$. A {\em transmission policy} refers to the pair of power allocation functions $(p_s(t), p_r(t))$.

A transmission policy should satisfy energy and data causality constraints in order to be feasible. The feasible set of transmission policies can be represented as
\vspace{-0.05in}
\begin{eqnarray}
\label{eq 2}
\mathfrak{F}&=&\{p_s(t),p_r(t)|p_s(t)\geq 0, p_r(t)\geq 0, p_s(t)p_r(t)=0,\nonumber \\
            &&E_s(t)\leq\sum_{i:0\leq t_{s,i}<t}^{}{E_{s,i}}, E_r(t)\leq\sum_{i:0\leq t_{r,i}<t}{E_{r,i}}, \nonumber \\
            &&B_r(t)\leq B_s(t), t\in[0,T]\} \nonumber
\vspace{-0.05in}
\end{eqnarray}
Our goal is to maximize the throughput as in \cite{Yener}, which is equivalent to maximizing the total transmitted data from the source to the destination until the deadline \emph{T}. The corresponding optimization problem can be formulated as
\begin{eqnarray}
\label{prob 1}
\text{max}\; \; B_r(T)\; \; \;  \text{subject to}\;\; (p_s(t), p_r(t))\in \mathfrak{F}
\end{eqnarray}

In the next section, we will provide some basic properties of the optimal transmission policy which will then be used in Section \ref{optimization of specific cases} to derive an optimal policy for $M=N=2$.
\vspace{-0.1in}
\section{Properties of Optimal Transmission Policy}
\label{Properties of Optimal Scheduling and Power Allocation}
\begin{lemma}
\label{lemma 1}
In a single epoch, constant power transmission is optimal between energy harvests.
\end{lemma}
\begin{proof}
This follows from concavity of $g(p)$ and by Jensen's inequality \cite{info theory}; an explicit proof is given in \cite[Lemma 2]{Yang2010} for single link system.
\end{proof}

\begin{rem}
\label{remark 1}
Lemma 1 suggests that without loss of generality we can consider power allocation functions $p_s(t)$ and $p_{r}(t)$ that consist of a sequence of nonzero constant power levels $p_{s,i}$ and $p_{r,j}$, $i=1,...,m$, $j=1,...,n$, respectively. We assume that the power level $p_{s,i}$ $(p_{r,j})$ has duration $l_{s,i}$ $(l_{r,j})$.
\end{rem}

\begin{lemma}
\label{lemma 2}
It is optimal to deplete all batteries by the deadline T.
\end{lemma}
\begin{proof}
This is proved in \cite[Lemma 5]{Yener} for a single energy harvesting link with finite capacity battery. The proof can be simply extended to our model.
\end{proof}

The following remarks will be useful in proving the subsequent lemmas.

\begin{rem}
\label{remark 2}
Given any feasible transmission policy, delaying an epoch of the relay does not violate feasibility (provided half-duplex constraint is still satisfied) because postponing the relay transmission allows it to store more energy. Similarly, moving an epoch of the source to an earlier time (subject to half-duplex and source energy causality constraints) does not violate data causality. Moreover, the above argument still holds if an epoch of the relay is delayed to switch with a later epoch of the source.
\end{rem}

\begin{rem}
\label{remark 3}
Consider an optimal transmission policy. We can obtain another feasible policy by moving all source epochs to earlier time instants (possibly by delaying relay epochs) provided energy causality of the source is maintained. Using Remark \ref{remark 2} we can argue that the new transmission policy maintains optimality. Hence, we will restrict attention to policies for which the source transmits whenever it has energy.
\end{rem}

\begin{lemma}
\label{lemma 3}
Given any feasible transmission policy for which $p_s(t)=p_r(t)=0$ for some $t\in[0,T]$, we can find another feasible policy such that it transmits at least as much data by fully utilizing channel, that is never turning off both the source and the relay.
\end{lemma}
\begin{proof}
This is argued in \cite{antepli} for a single link using properties of $g(p)$. To extend it to the two-hop case we define an idle time as a period for which $p_s(t)=p_r(t)=0$. For any feasible transmission policy, if we have an idle time right after an epoch of the source, we can invoke arguments similar to \cite{antepli} to extend the epoch of the source and ensure that there is no idle time. Note that this strategy continues to satisfy energy and data causality constraints. On the other hand, if an idle time occurs after an epoch of the relay preceded by a source epoch, the relay epoch can de delayed without violating feasibility (see Remark \ref{remark 2}), and the previous argument can be used to extend the source epoch to result in no idle time. For the case of multiple consecutive relay epochs this argument can be applied repeatedly to each epoch.
\end{proof}

\begin{cor}
\label{Corollary 1}
In an optimal transmission policy, the following properties are satisfied:
\begin{enumerate}
\item Epoches that belong to same node are not adjacent, that is, each epoch of the source is followed by an epoch of the relay.
\item Source and relay batteries cannot be empty simultaneously for any $t\in[0,T]$.
\item A new epoch of the source starts immediately after the data buffer of the relay becomes empty.
\end{enumerate}
\end{cor}

\begin{lemma}
\label{lemma 4}
In an optimal transmission policy, source and relay power levels form a non-decreasing sequence in time that is $p_{s,1}\leq p_{s,2}\leq...\leq p_{s,m}$ and $p_{r,1}\leq p_{r,2}\leq...\leq p_{r,n}$.
\end{lemma}
\begin{proof}
First, we consider source power levels within an epoch. For any transmission policy, keeping the consumed energy and duration of each source epoch the same, the source-relay communication in that epoch can be thought of as a single energy harvesting link. Using \cite[Lemma 1]{Yang2010} where transmission time minimization problem is considered, we can argue that by having monotonically increasing power levels in that epoch, the duration of the epoch can be decreased while the amount of data transmitted and total energy consumed in that epoch remains the same. This, however, leads to an idle time, which can be removed by Lemma \ref{lemma 3}. The same argument is valid for the relay as well.

Next, we consider source power levels in successive epochs. Consider power levels $p_{s,i}$ and $p_{s,i+1}$ in two successive epochs of the source (recall that zero power levels are not considered in $p_{s,i}$). Assume $p_{s,i}>p_{s,i+1}$. Then, by keeping $l_{s,i}$ and $l_{s,i+1}$ the same and equalizing power levels, the source transmits strictly more total data due to strict concavity of $g(p)$. However, the transmitted data in the first epoch decreases due to decrease in $p_{s,i}$. Therefore, the new policy may violate data causality in the following relay epoch. By delaying the relay epoch (Remark \ref{remark 2}), we can find another transmission policy which does not violate data causality. Similarly, for two successive relay epochs assuming $p_{r,i}>p_{r,i+1}$, the relay transmits strictly more data by equalizing power levels. Unlike source, decrement of transmitted data in the first epoch is always feasible. However, increase in the total transmitted data by the relay may violate data causality in the following epochs. In that case, the new policy can be further replaced by another one which is obtained by increasing duration of preceding source epoch while keeping consumed energy the same and decreasing duration of the first relay epoch while keeping power levels equal. Clearly, the new policy transmits strictly more data. Finally note that after equalizing power levels of successive epochs, if necessary, power levels in each epoch can be further arranged to ensure that power does not decrease.
\end{proof}

\begin{lemma}
\label{lemma 5}
Any feasible transmission policy can be replaced with one that satisfies the following property without decreasing the data rate: Within an epoch of source (relay), whenever the source (relay) power changes, total consumed source (relay) energy up to that point becomes equal to total harvested energy.
\end{lemma}
\begin{proof}
Suppose for a feasible transmission policy, the statement of the lemma does not hold for a particular source epoch. Then by \cite{Yang2010}, the policy can be replaced by another one which transmits the same amount of data consuming the same energy but in a shorter epoch duration. Note that by \cite[Lemma 3]{Yang2010} this new policy satisfies the property specified in the Lemma. Then by Lemma \ref{lemma 3}, this policy can be further replaced by another one which transmits at least the same amount of data and leaves no idle time. The same argument is valid for the relay as well.
\end{proof}

\begin{lemma}
\label{lemma 6}
Any feasible transmission policy can be replaced with one that satisfies the following property without decreasing the data rate: There is a relay epoch between source power levels $p_{s,i}<p_{s,i+1}$ provided there is nonzero energy in relay battery just before the time that source power level changes.
\end{lemma}
\begin{proof}
Without loss of generality, consider an epoch of the source in which the statement of the Lemma does not hold. Also assume this epoch has two power levels $p_{s,i}<p_{s,i+1}$ such that total consumed source energy up to the time that the power level changes equal to total harvested source energy (Lemma \ref{lemma 5}). By keeping consumed energy for the duration of each power level ($l_{s,i}$ and $l_{s,i+1}$) the same, the policy can be replaced by another one which has power levels $p_{s,i}'$, $p_{s,i+1}'$ of durations $l_{s,i}'=l_{s,i}-\epsilon$, $l_{s,i+1}'=l_{s,i+1}+\epsilon$ with $p_{s,i}<p_{s,i}'<p_{s,i+1}'<p_{s,i+1}$ and relay transmitting for $\epsilon$ duration in between these power levels provided it has energy. Relay transmission is obtained by moving $\epsilon$ portion of the following relay epoch earlier. Since relay has nonzero data in its data buffer, the new policy remains feasible. In the new policy, difference between the first and second power levels $p_{s,i}'$ and $p_{s,i+1}'$ decreases without changing total transmission duration of the nodes, hence, the source transmits more data due to strict concavity of $g(p)$.
\end{proof}

\begin{lemma}
\label{lemma 9}
In an optimal transmission policy, whenever the relay power levels change in two successive epochs, either the data or the energy buffer of the relay is empty at the end of the first epoch.
\end{lemma}
\begin{proof}
The proof is given by contradiction. Suppose the lemma is not satisfied. Consider power levels $p_{r,i}<p_{r,i+1}$ in two successive epochs of relay (see Lemma \ref{lemma 4}). Then, by keeping $l_{r,i}$ and $l_{r,i+1}$ the same we can increase $p_{r,i}$ and decrease $p_{r,i+1}$ unless battery or data buffer of the relay is empty at the end of the first epoch. As a result, the relay transmits strictly more data due to strict concavity of the $g(p)$. This operation can be done until battery or data buffer depletes, or power levels equalize. However, due to increase in the total delivered data by the relay, the new transmission policy may violate data causality in the subsequent epochs. In that case, as in the proof of Lemma \ref{lemma 4} this policy can be further replaced by another one which is obtained by decreasing duration of the first relay epoch, and increasing duration of the preceding source epoch until data causality is satisfied. Clearly, in the new transmission policy, the total transmitted data by the relay is increased compared to the initial policy. This results in a contradiction, hence the lemma must be true.
\end{proof}

\begin{lemma}
\label{lemma 8}
In an optimal transmission policy, source and relay transmit same amount of data until the deadline.
\end{lemma}
\begin{proof}
Consider a feasible transmission policy such that the source transmits more data than the relay until the deadline, i.e., $B_r(T)<B_s(T)$, hence total delivered data to the destination is $B_r(T)$. While keeping the consumed energy in each epoch the same, increasing duration of an epoch strictly increases transmitted data in that epoch \cite{antepli}. Therefore, the initial policy can be replaced by another one of higher rate obtained by increasing duration of relay epochs while decreasing duration of the following source epochs under data causality constraint. However, if the first epoch of the source transmits more data than total transmitted data in the following relay epochs, we need to decrease the duration of the first source epoch while increasing duration of the first relay epoch. Note that increasing duration of the first relay epoch is always feasible due to $E_{r,0}>0$. Combining these, we can find another feasible policy transmitting higher data such that source and relay transmit same amount of data until the deadline, i.e., $B_r(T)=B_s(T)$.
\end{proof}

\vspace{-0.1in}
\section{Optimal Policy for Two Energy Arrival at the Source and the relay}
\label{optimization of specific cases}
In this section, using the properties developed in Section \ref{Properties of Optimal Scheduling and Power Allocation}, we describe the optimal transmission policy for the two energy arrival case, $M=N=2$. We use the {\em cumulative curve} model as in \cite{zafer}. We will refer to cumulative energy arrival curve as the {\em harvested energy curve} \cite{deniz2}. Similarly, the {\em transmitted energy curve} refers to $E_s(t)$ or $E_r(t)$. By Remark \ref{remark 1} we express $p_s(t)$ by sequence of constant source power levels including zero power using the power vector $\mathbf{p}_s=[p_{s,1},0,p_{s,2},...]$ with the corresponding duration vector $\mathbf{l}_s=[l_{s,1},\xi_{r,1},l_{s,2},...]$. Similarly, $\mathbf{p}_r$ and $\mathbf{l}_r$ can be defined. Note that ($\mathbf{p}_s,\mathbf{l}_s,\mathbf{p}_r,\mathbf{l}_r$) constitutes a transmission policy.

The solution of the optimization problem (\ref{prob 1}) depends on the location of the {\em corner points} $(t_{s,1},E_{s,0})$ and $(t_{r,1},E_{r,0})$ of the harvested energy curves of the source and the relay, respectively, in the energy-time region. Here, $t_{s,1}$ and $t_{r,1}$ are the source and relay energy arrival instants, $E_{s,0}$ and $E_{r,0}$ are initial battery energies and the {\em energy-time region} of the source and the relay refer to the rectangles defined by the points $(T,E_{s,0}+E_{s,1})$ and $(T,E_{r,0}+E_{r,1})$, respectively. We first identify two sub-regions $\mathcal{S}_1$ and $\mathcal{S}_2$ of the energy-time region of the source as shown in Figure \ref{fig 3:subfig1}. Boundaries of these regions are specified by dashed lines, and will be described subsequently. Then, based on the location of $(t_{s,1},E_{s,0})$ with respect to these sub-regions we identify regions on the energy-time region of the relay. If $(t_{s,1},E_{s,0})\in \mathcal{S}_1$, there is a single relay region which is the complete relay energy-time region. If $(t_{s,1},E_{s,0})\in\mathcal{S}_2$ we specify four sub-regions {$\mathcal{R}_1, \mathcal{R}_2, \mathcal{R}_3, \mathcal{R}_4$} for the relay corner point as illustrated in Figure \ref{fig 3:subfig2}. These sub-regions depend on the harvested energy curve of the source as well as $(T,E_{r,0}+E_{r,1})$. Boundaries of each region along with the respective optimal transmission policy will be described when the corresponding region is discussed.

In the solutions described, it can be argued that obtained optimal power vectors $\mathbf{p}_s$ and $\mathbf{p}_r$ are unique. The proof of uniqueness will be omitted due to space constraints.
\begin{figure}[t]
\centering
\subfigure[]{
\includegraphics[scale=1,trim= 40 5 5 0]{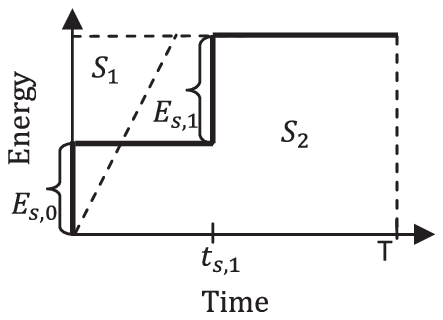}
\label{fig 3:subfig1}
\vspace{-0.2in}
}
\subfigure[]{
\includegraphics[scale=1,trim= 10 5 30 0]{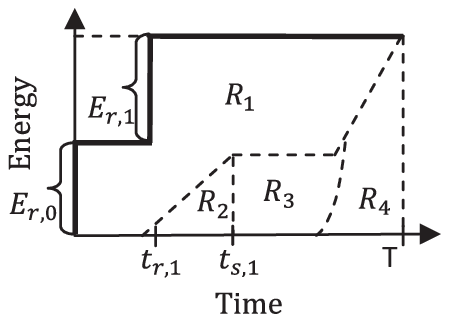}
\label{fig 3:subfig2}
}
\vspace{-0.2cm}
\caption{(a) Source energy-time region and sub-regions for the corner point $(t_{s,1},E_{s,0})$. (b) Relay energy-time region and sub-regions for the corner point $(t_{r,1},E_{r,0})$. Cumulative harvested energy is represented by solid lines and boundaries of the regions are represented by dashed lines.}
\label{fig 3}%
\vspace{-0.2in}
\end{figure}
\vspace{-0.15in}
\subsection{Source Region $\mathcal{S}_1$}
Source region $\mathcal{S}_1$ is defined such that the source will be able to transmit with constant power in a single epoch without violating energy constraints (Lemma \ref{lemma 1}, Remark \ref{remark 3}), and the relay will transmit in the remaining time (Lemma \ref{lemma 3}). Thus, the problem reduces to single energy arrival for the source and two energy arrival for the relay, i.e., $M=1$, $N=2$. This problem is solved in \cite{deniz} which gives a unique optimal source epoch length $\xi_{s,1}^*=l_{s,1}$. Clearly, the source transmitted energy curve given by the line starting from origin with slope $p_{s,1}=\frac{E_{s,0}+E_{s,1}}{\xi_{s,1}^*}$ provides the boundary of $\mathcal{S}_1$ because the source can deplete all harvested energy in a single epoch if and only if the corner point $(t_{s,1},E_{s,0})$ is above this line (Lemma \ref{lemma 5}). Therefore the condition $\frac{E_{s,0}+E_{s,1}}{\xi_{s,1}^*}\leq\frac{E_{s,0}}{t_{s,1}}$ defines region $\mathcal{S}_1$, and the algorithm in \cite{deniz} provides an optimal transmission policy.
\vspace{-0.15in}
\subsection{Source Region $\mathcal{S}_2$}
Whenever $\frac{E_{s,0}}{t_{s,1}}<\frac{E_{s,0}+E_{s,1}}{\xi_{s,1}^*}$, the corner point $(t_{s,1},E_{s,0})$ will be in region $\mathcal{S}_2$.
In this region, the source has two epochs (Lemma \ref{lemma 6}), hence, the relay has two epochs (Lemma \ref{lemma 3}, Remark \ref{remark 3}). Also, the source has constant power levels $p_{s,1}$ and $p_{s,2}$ with lengths $\xi_{s,1}$ and $\xi_{s,2}$, respectively, in these epochs. The optimal transmission policy further depends on the region relay's corner point $(t_{r,1},E_{r,0})$ lies in and will be specified below.

\subsubsection{Relay Region $\mathcal{R}_1$}
In this sub-region, relay's harvested energy does not restrict its transmitted energy curve. Therefore we argue that total transmitted data is equal to the case that the relay harvests all its total energy $E_{r,0}+E_{r,1}$ at $t=0$. Using Lemma \ref{lemma 1} and Lemma \ref{lemma 5}, both the source and the relay have a single power level in each epoch. We first find an optimal transmission policy by assuming that the source harvests all its energy $E_{s,0}+E_{s,1}$ at $t=0$. Since this corresponds to largest transmitted data, if we can find a feasible policy transmitting this amount of data, it needs to be optimal. Solving
\begin{eqnarray}
\label{single arrival}
& g\left(\frac{E_{s,0}+E_{s,1}}{t}\right)t=g\left(\frac{E_{r,0}+E_{r,1}}{T-t}\right)(T-t)
\vspace{-0.05in}
\end{eqnarray}
for $t$ gives the optimal division of total time $T$ between the source and the relay when all energies are harvested at $t=0$. Equality in (\ref{single arrival}) is due to Lemma \ref{lemma 8} ($B_s(T)=B_r(T)=B(T)$). To obtain a transmission policy we set $\xi_{s,1}+\xi_{s,2}=t$ where t is the solution to (2) and use power levels as in the solution of (\ref{single arrival}). The transmitted data in the first epoch of the source and the relay are then $\frac{\xi_{s,1}}{t}B(T)$ and $\frac{t_{s,1}-\xi_{s,1}}{T-t}B(T)$, respectively (Remark \ref{remark 3}). If $\frac{\xi_{s,1}}{t}B(T)\geq \frac{t_{s,1}-\xi_{s,1}}{T-t}B(T)$, the transmission policy is feasible. Then, the optimal transmission policy is $\mathbf{p}_s=[\frac{E_{s,0}+E_{s,1}}{t},0,\frac{E_{s,0}+E_{s,1}}{t},0]$, $\mathbf{l}_s=[\frac{E_{s,0}}{E_{s,0}+ E_{s,1}}t,t_{s,1}-\frac{E_{s,0}}{E_{s,0}+ E_{s,1}}t,\frac{E_{s,1}}{E_{s,0}+ E_{s,1}}t,\tau_{s,2}-\frac{E_{s,1}}{E_{s,0}+ E_{s,1}}t]$, $\mathbf{p}_r=[0,\frac{E_{r,0}+E_{r,1}}{T-t},0,\frac{E_{r,0}+E_{r,1}}{T-t}]$, $\mathbf{l}_r=[\frac{E_{s,0}}{E_{s,0}+ E_{s,1}}t,t_{s,1}-\frac{E_{s,0}}{E_{s,0}+ E_{s,1}}t,\frac{E_{s,1}}{E_{s,0}+ E_{s,1}}t,\tau_{s,2}-\frac{E_{s,1}}{E_{s,0}+ E_{s,1}}t]$, where $t$ is the solution of (\ref{single arrival}).

If above transmission policy violates data causality constraint, i.e., $\frac{\xi_{s,1}}{t}B(T)< \frac{t_{s,1}-\xi_{s,1}}{T-t}B(T)$, invoking Lemma \ref{lemma 8}, the optimization problem in ({\ref{prob 1}}) can instead be written as
\vspace{-0.1in}
\begin{eqnarray}
\label{ob 1}
\lefteqn{\underset{\xi_{s,1},\xi_{s,2}}{\text{max}}\sum_{i=1}^2{g\left(\frac{E_{s,i-1}}{\xi_{s,i}}\right)\xi_{s,i}}}\\
\label{op 7:1}
&\text{s.t.}\; \; g\left(\frac{E_{s,i-1}}{\xi_{s,i}}\right)\xi_{s,i} =g\left(\frac{E_{r,i-1}'}{\tau_{s,i}-\xi_{s,i}}\right)(\tau_{s,i}-\xi_{s,i})\ i=1,2
\vspace{-0.05in}
\end{eqnarray}
where $\xi_{s,1}<\tau_{s,1}=t_{s,1}$, $\xi_{s,2}<\tau_{s,2}$, $0<E_{r,0}'<E_{r,0}+E_{r,1}$ and $E_{r,0}'+E_{r,1}'=E_{r,0}+E_{r,1}$. The equality in (\ref{op 7:1}) is due to empty data buffer at the end of the first relay epoch. Then, the optimal transmission policy is $\mathbf{p}_s=[\frac{E_{s,0}}{\xi_{s,1}},0,\frac{E_{s,1}}{\xi_{s,2}},0]$, $\mathbf{l}_s=[\xi_{s,1},t_{s,1}-\xi_{s,1},\xi_{s,2},\tau_{s,2}-\xi_{s,2}]$, $\mathbf{p}_r=[0,\frac{E_{r,0}'}{t_{s,1}-\xi_{s,1}},0,\frac{E_{r,1}'}{\tau_{s,2}-\xi_{s,2}}]$, $\mathbf{l}_r=[\xi_{s,1},t_{s,1}-\xi_{s,1},\xi_{s,2},\tau_{s,2}-\xi_{s,2}]$, where $\xi_{s,1}$, $\xi_{s,2}$ are the solutions to (\ref{ob 1}) subject to (\ref{op 7:1}).

Clearly, the relay transmitted energy curve given by the optimal transmission policy ($\mathbf{p}_r,\mathbf{l}_r$) above provides the boundary between $\mathcal{R}_1$ and the other relay regions because if $(t_{r,1},E_{r,0})$ is above this curve, the policy ($\mathbf{p}_r,\mathbf{l}_r$) is feasible.

\subsubsection{Relay Region $\mathcal{R}_2$}
In this sub-region, energy causality is violated in the first epoch of the relay if the transmission policy in $\mathcal{R}_1$ is used. Therefore, using Remark \ref{remark 3}, Lemma \ref{lemma 3} and Lemma \ref{lemma 5}, the relay power level changes at $t_{r,1}$ resulting in two power levels $p_{r,1}$ and $p_{r,2}$ in the first epoch of the relay, and a third power level $p_{r,3}$ in the second relay epoch. Note that $p_{r,2}$ and $p_{r,3}$ are different only if the relay data buffer is empty at the end of the first epoch (Lemma \ref{lemma 9}). In the optimal policy whether the data buffer is empty or not is not known in advance; therefore we provide two different solutions based on the state of the relay data buffer. When $p_{r,2}=p_{r,3}$, we maximize (\ref{ob 1}) subject to
\vspace{-0.05in}
\begin{eqnarray}
\label{op 2}
&\sum_{i=1}^2{g\left(\frac{E_{s,i-1}}{\xi_{s,i}}\right)\xi_{s,i}}=\sum_{i=1}^2{g\left(\frac{E_{r,i-1}}{\tau_{r,i}-\xi_{s,i}}\right)(\tau_{r,i}-\xi_{s,i})}
\vspace{-0.05in}
\end{eqnarray}
where $\xi_{s,1}<\tau_{r,1}=t_{r,1}$, $\xi_{s,2}<\tau_{s,2}$. Note that, $l_{r,2}+l_{r,3}$ are considered together as $\tau_{r,2}-\xi_{s,2}$ because $p_{r,2}=p_{r,3}$, and $l_{r,1}=t_{r,1}-\xi_{s,1}$. The optimal policy is $\mathbf{p}_s=[\frac{E_{s,0}}{\xi_{s,1}},0,\frac{E_{s,1}}{\xi_{s,2}},0]$, $\mathbf{l}_s=[\xi_{s,1},t_{s,1}-\xi_{s,1},\xi_{s,2},\tau_{s,2}-\xi_{s,2}]$, $\mathbf{p}_r=[0,\frac{E_{r,0}}{t_{r,1}-\xi_{s,1}},\frac{E_{r,1}}{\tau_{r,2}-\xi_{s,2}},0,\frac{E_{r,1}}{\tau_{r,2}-\xi_{s,2}}]$,  $\mathbf{l}_r=[\xi_{s,1},t_{r,1}-\xi_{s,1},t_{s,1}-t_{r,1},\xi_{s,2},\tau_{s,2}-\xi_{s,2}]$ where $\xi_{s,1}$, $\xi_{s,2}$ are the solutions to (\ref{ob 1})-(\ref{op 2}).

If the above transmission policy violates data causality, i.e., $g\left(p_{s,1}\right)\xi_{s,1}<g\left(p_{r,1}\right)(t_{r,1}-\xi_{s,1})+g\left(p_{r,2}\right) (t_{s,1}-t_{r,1})$ in the above policy, we maximize (\ref{ob 1}) subject to
\vspace{-0.05in}
\begin{eqnarray}
\label{op 3:1}
\lefteqn{g\left(\frac{E_{s,0}}{\xi_{s,1}}\right)\xi_{s,1}}\\
   & =g\left(\frac{E_{r,0}}{t_{r,1}-\xi_{s,1}}\right)(t_{r,1}-\xi_{s,1})+g\left(\frac{\hat{E}_{r,1}}{t_{s,1}-t_{r,1}}\right)(t_{s,1}-t_{r,1})\nonumber\\
\label{op 3:2}
\lefteqn{g\left(\frac{E_{s,1}}{\xi_{s,2}}\right)\xi_{s,2}=g\left(\frac{E_{r,1}-\hat{E}_{r,1}}{\tau_{s,2}-\xi_{s,2}}\right)(\tau_{s,2}-\xi_{s,2})}
\vspace{-0.05in}
\end{eqnarray}
where $\xi_{s,1}<t_{r,1}$, $\xi_{s,2}<\tau_{s,2}$ and $\hat{E}_{r,1}<E_{r,1}$ with $\hat{E}_{r,1}$ corresponding to the consumed energy by the relay in the interval $(t_{r,1}, t_{s,1})$. The first constraint in (\ref{op 3:1}) is due to empty data buffer at the end of the first relay epoch, and the second constraint in (\ref{op 3:2}) is due to Lemma \ref{lemma 8}. The resulting optimal transmission policy is $\mathbf{p}_s=[\frac{E_{s,0}}{\xi_{s,1}},0,\frac{E_{s,1}}{\xi_{s,2}},0]$, $\mathbf{l}_s=[\xi_{s,1},t_{s,1}-\xi_{s,1},\xi_{s,2},\tau_{s,2}-\xi_{s,2}]$, $\mathbf{p}_r=[0,\frac{E_{r,0}}{t_{r,1}-\xi_{s,1}},\frac{\hat{E}_{r,1}}{t_{s,1}-t_{r,1}},0,\frac{E_{r,1}-\hat{E}_{r,1}}{\tau_{s,2}-\xi_{s,2}}]$, $\mathbf{l}_r=[\xi_{s,1},t_{r,1}-\xi_{s,1},t_{s,1}-t_{r,1},\xi_{s,2},\tau_{s,2}-\xi_{s,2}]$ where $\xi_{s,1}$, $\xi_{s,2}$ and $\hat{E}_{r,1}$ are the solutions to (\ref{ob 1}) subject to (\ref{op 3:1})-(\ref{op 3:2}). The boundary between $\mathcal{R}_2$ and $\mathcal{R}_3$ is specified as the vertical line at $t_{s,1}$ (Remark \ref{remark 3}).

\subsubsection{Relay Region $\mathcal{R}_3$}
We define this sub-region as the region in which the relay optimal transmission policy only uses the first harvested energy $E_{r,0}$ in its first epoch and the second harvested energy $E_{r,1}$ in the second epoch (Remark \ref{remark 3}, Lemma \ref{lemma 9}). Therefore, the relay has power levels $p_{r,1}=\frac{E_{r,0}}{\xi_{r,1}}$ and $p_{r,2}=\frac{E_{r,1}}{\xi_{r,2}}$ with durations $\xi_{r,1}$ and $\xi_{r,2}$ in the first and the second epoch of the relay, respectively. As discussed in Lemma \ref{lemma 9}, these power levels are different since all energy $E_{r,0}$ is used in the first epoch. In the following lemma we prove that relay still has data in its buffer at the end of the first epoch.
\begin{lemma}
\label{lemma 11}
In an optimal transmission policy, relay has nonzero data at the end of the first epoch if $(t_{r,1},E_{r,0})\in\mathcal{R}_3$.
\end{lemma}
\begin{proof}
First consider a policy of the form described above in which the relay consumes $E_{r,0}$ and has empty data buffer at the end of its first epoch. Let's denote corresponding optimal epoch durations of the source as $(\xi_{s,1}^1,\xi_{s,2}^1)$. Now, consider the optimal policy in $\mathcal{R}_1$. Clearly, in that policy consumed energy in the first epoch of the relay is larger than $E_{r,0}$. Let's denote corresponding epoch durations of the source as $(\xi_{s,1}^2,\xi_{s,2}^2)$. Clearly, $\xi_{s,1}^2>\xi_{s,1}^1$ and $\xi_{s,2}^1>\xi_{s,2}^2$ due to concavity of $g(p)$. Then, in order to make the second policy consume $E_{r,0}$ in the first epoch of the relay, we can increase first epoch duration $\xi_{s,1}^2$ and decrease second epoch duration $\xi_{s,2}^2$ of the source. We denote modified epoch durations of the second policy as $(\xi_{s,1}^3,\xi_{s,2}^3)$. Then, we have two feasible policies both consuming $E_{r,0}$ in the first relay epoch; one has source epoch durations $(\xi_{s,1}^1,\xi_{s,2}^1)$ such that data buffer of the relay is empty at the end of the first relay epoch, the other one has source epoch durations $(\xi_{s,1}^3, \xi_{s,2}^3)$. Consider the policy with epoch durations $(\xi_{s,1}',\xi_{s,1}')$ that are convex combinations of $(\xi_{s,1}^1,\xi_{s,2}^1)$ and $(\xi_{s,1}^3,\xi_{s,2}^3)$. This policy satisfies energy causality and transmits strictly more data due to strict concavity of the objective function in (\ref{ob 1}). Since $\xi_{s,1}'>\xi_{s,1}^3$, it is also feasible, and the relay has nonzero data at the end of its first epoch.
\end{proof}

Using Lemma \ref{lemma 11}, we maximize (\ref{ob 1}) subject to
\vspace{-0.05in}
\begin{eqnarray}
\label{op 4:1}
&\sum_{i=1}^2{g\left(\frac{E_{s,i-1}}{\xi_{s,i}}\right)\xi_{s,i}}=\sum_{i=1}^2{g\left(\frac{E_{r,i-1}}{\tau_{s,i}-\xi_{s,i}}\right)(\tau_{s,i}-\xi_{s,i})}\\
\label{op 4:2}
&t_{r,1}-t_{s,1}\leq \xi_{s,2}
\vspace{-0.05in}
\end{eqnarray}
where $\xi_{s,1}<t_{s,1}$, $\xi_{s,2}<\tau_{s,2}$. The constraint in (\ref{op 4:2}) is due to definition of $\mathcal{R}_3$. The optimal policy is $\mathbf{p}_s=[\frac{E_{s,0}}{\xi_{s,1}},0,\frac{E_{s,1}}{\xi_{s,2}},0]$, $\mathbf{l}_s=[\xi_{s,1},t_{s,1}-\xi_{s,1},\xi_{s,2},\tau_{s,2}-\xi_{s,2}]$, $\mathbf{p}_r=[0,\frac{E_{r,0}}{t_{s,1}-\xi_{s,1}},0,\frac{E_{r,1}}{\tau_{s,2}-\xi_{s,2}}]$, $\mathbf{l}_r=[\xi_{s,1},t_{s,1}-\xi_{s,1},\xi_{s,2},\tau_{s,2}-\xi_{s,2}]$ where $\xi_{s,1}$, $\xi_{s,2}$ are the solutions to (\ref{ob 1}) subject to (\ref{op 4:1})-(\ref{op 4:2}).
\subsubsection{Relay Region $\mathcal{R}_4$}
In this sub-region, energy causality is violated in the second epoch of the relay if the transmission policy in $\mathcal{R}_1$ is used. Therefore, the relay power level changes at $t_{r,1}$ as argued in Lemma \ref{lemma 5}. Hence there is a single power level $p_{r,1}$ in the first relay epoch and there are two power levels $p_{r,2}$ and $p_{r,3}$ in the second epoch. Note that $p_{s,1}$ and $p_{s,2}$ are different only if the relay data buffer is empty at the end of the first epoch (Lemma \ref{lemma 9}). However, whether the data buffer is empty or not is not known in advance. Hence we first solve the case that relay data buffer is not empty. Solving
\vspace{-0.05in}
\begin{eqnarray}
\label{eq 1}
&g\left(\frac{E_{s,0}+E_{s,1}}{t}\right)t=g\left(\frac{E_{r,0}}{t_{r,1}-t}\right)(t_{r,1}-t)+g\left(\frac{E_{r,1}}{\tau_{r,2}}\right)\tau_{r,2}
\vspace{-0.05in}
\end{eqnarray}
for $t=\xi_{s,1}+\xi_{s,2}$ gives the optimal total transmission duration of the source. We further have $t_{r,1}-t=l_{r,1}+l_{r,2}$ because $p_{r,1}=p_{r,2}$, and $\xi_{s,2}<t_{r,1}-t_{s,1}$ (due to definition of $\mathcal{R}_4$). The resulting optimal policy is $\mathbf{p}_s=[\frac{E_{s,0}+E_{s,1}}{t},0,\frac{E_{s,0}+E_{s,1}}{t},0]$, $\mathbf{l}_s=[\frac{E_{s,0}}{E_{s,0}+E_{s,1}}t,t_{s,1}-\frac{E_{s,0}}{E_{s,0}+E_{s,1}}t,\frac{E_{s,1}}{E_{s,0}+E_{s,1}}t,\tau_{s,2}-\frac{E_{s,1}}{E_{s,0}+E_{s,1}}t]$, $\mathbf{p}_r=[0,\frac{E_{r,0}}{t_{r,1}-t}, 0,\frac{E_{r,0}}{t_{r,1}-t},\frac{E_{r,1}}{\tau_{r,2}}]$, $\mathbf{l}_r=[\frac{E_{s,0}}{E_{s,0}+E_{s,1}}t, t_{s,1}-\frac{E_{s,0}}{E_{s,0}+E_{s,1}}t, \frac{E_{s,1}}{E_{s,0}+E_{s,1}}t,t_{r,1}-t_{s,1}-\frac{E_{s,1}}{E_{s,0}+E_{s,1}}t, \tau_{r,2}]$ where $t$ is the solution of (\ref{eq 1}).

If the above transmission policy is infeasible, i.e., $g\left(p_{s,1}\right)\xi_{s,1}<g\left(p_{r,1}\right)(t_{s,1}-\xi_{s,1})$, then we maximize (\ref{ob 1}) subject to
\vspace{-0.05in}
\begin{eqnarray}
\label{op 5:1}
\lefteqn{g\left(\frac{E_{s,0}}{\xi_{s,1}}\right)\xi_{s,1}=g\left(\frac{\hat{E}_{r,0}}{t_{s,1}-\xi_{s,1}}\right)(t_{s,1}-\xi_{s,1})}\\
\label{op 5:2}
\lefteqn{g\left(\frac{E_{s,1}}{\xi_{s,2}}\right)\xi_{s,2}}\nonumber\\
& =g\left(\frac{E_{r,0}-\hat{E}_{r,0}}{t_{r,1}-t_{s,1}-\xi_{s,2}}\right)(t_{r,1}-t_{s,1}-\xi_{s,2})+g\left(\frac{E_{r,1}}{\tau_{r,2}}\right)\tau_{r,2}
\vspace{-0.05in}
\end{eqnarray}
where $\xi_{s,1}<t_{s,1}$, $\xi_{s,2}<t_{r,1}-t_{s,1}$ and $\hat{E}_{r,0}<E_{r,0}$ with $\hat{E}_{r,0}$ corresponding to the consumed relay energy in interval $(\xi_{s,1}, t_{s,1})$. The resulting optimal policy is $\mathbf{p}_s=[\frac{E_{s,0}}{\xi_{s,1}},0,\frac{E_{s,1}}{\xi_{s,2}},0]$, $\mathbf{l}_s=[\xi_{s,1},t_{s,1}-\xi_{s,1},\xi_{s,2},\tau_{s,2}-\xi_{s,2}]$, $\mathbf{p}_r=[0,\frac{\hat{E}_{r,0}}{t_{s,1}-\xi_{s,1}},0,\frac{E_{r,0}-\hat{E}_{r,0}}{t_{r,1}-t_{s,1}-\xi_{s,2}},\frac{E_{r,1}}{\tau_{r,2}}]$, $\mathbf{l}_r=[\xi_{s,1},t_{s,1}-\xi_{s,1},\xi_{s,2},t_{r,1}-t_{s,1}-\xi_{s,2},\tau_{r,2}]$ where $\xi_{s,1}$, $\xi_{s,2}$ and $\hat{E}_{r,0}$ are the solutions to (\ref{ob 1}) subject to (\ref{op 5:1})-(\ref{op 5:2}).

There is no closed form expression for the boundary between $\mathcal{R}_3$ and $\mathcal{R}_4$, therefore to decide which relay region a corner point $(t_{r,1},E_{r,0})$ lies in, we check the following conditions:
\begin{itemize}
\item If $g(\frac{E_{s,1}}{t_{r,1}-t_{s,1}})(t_{r,1}-t_{s,1})> g(\frac{E_{r,1}}{T-t_{r,1}})(T-t_{r,1})$, the corner point $(t_{r,1},E_{r,0})$ is in $\mathcal{R}_4$. This follows from constraint (\ref{op 4:2}), that is, if $(t_{r,1},E_{r,0})\in \mathcal{R}_3$, the total transmitted data in the second epoch of source would be more than the total transmitted data in the second epoch of relay violating (\ref{op 4:2}).
\item If $g(\frac{E_{s,1}}{t_{r,1}-t_{s,1}})(t_{r,1}-t_{s,1})\leq g(\frac{E_{r,1}}{T-t_{r,1}})(T-t_{r,1})$, maximizing (\ref{ob 1}) subject to (\ref{op 5:1})-(\ref{op 5:2}) cannot give a feasible policy because (\ref{op 5:1})-(\ref{op 5:2}) require empty data buffer at $t_{s,1}$ and $t_{r,1}-t_{s,1}>\xi_{s,2}$. In this case we first solve (\ref{eq 1}). If the resulting policy is not feasible, i.e. $g(p_{s,1})\xi_{s,1}>g(p_{r,1})(\tau_{s,1}-\xi_{s,1})$, then we conclude that $(t_{r,1},E_{r,0})\in \mathcal{R}_3$. It can be argued that when $(t_{r,1},E_{r,0})\in\mathcal{R}_3$, the solution of (\ref{eq 1}) never results in a feasible policy.
\end{itemize}

\vspace{-0.1in}
\section{Numerical and Simulation Results}

In this section, we provide simulation results to illustrate average performance improvement provided by the optimal transmission policy developed in Section \ref{optimization of specific cases} compared with a simple transmission policy. We set $T=1$. We assume that for both source and relay, harvested energies are independently chosen from the exponential distribution with parameter $\lambda$ and $t_{s,1}$ and $t_{r,1}$ are uniformly distributed in the interval [0,1). Once harvested energies and harvest instances are randomly chosen, they are revealed to source and the relay, which can then carry out the corresponding transmission policies. We consider the Shannon power-rate function $g(p)=\log(1+p)$. The simple transmission policy operates in a time slotted fashion with slot durations $\frac{T}{2}=0.5$. In the first slot, the source transmits and in the second slot the relay transmits. In each slot, the transmission powers are chosen to maximize the amount of data per hop using \cite{Yang2010}, and the total transmitted data is the minimum of transmitted data in each hop. Figure \ref{fig 4} shows the average throughput as a function of $\lambda$. This figure illustrates that significant throughput improvement is possible by employing the optimal transmission policy.

We next fix $E_{s,0}=E_{s,1}=E_{r,0}=E_{r,1}=5$, $T=10$, $t_{r,1}=8$ for $g(p)=\log(1+p)$. Figure \ref{fig 5} shows the optimal throughput (total data transmitted/T) and the corresponding regions for the source/relay corner points as a function of $t_{s,1}$. We observe that all parameters affect the location of the corner points and the resulting throughput.

\vspace{-0.1in}
\section{Conclusions}
In this paper we studied energy harvesting two-hop networks to maximize data delivered to the destination by a given deadline under non-causal knowledge of the harvested energy profiles. We first identified properties of an optimal transmission policy subject to energy and data causality constraints. Then, we provided the optimal policy for two energy arrival at the source and the relay. Numerical and simulation results clearly illustrate the benefits of employing an optimal policy. Details of the optimal policy for the multi-energy arrival case as well as online policies will be discussed in subsequent work.

\begin{figure}[t]
\center
\includegraphics[scale=0.75,trim= 0 5 0 0]{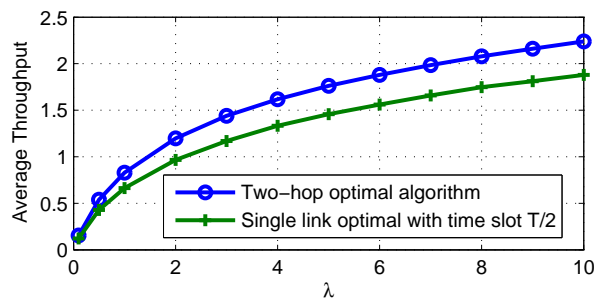}
\vspace{-0.1in}
\caption{Performance comparison of the optimal transmission  policy with a simple one that operates in a time-slotted fashion.}%
\label{fig 4}%
\vspace{-0.13in}
\end{figure}

\begin{figure}[t]
\center
\includegraphics[scale=0.75,trim= 0 15 0 5]{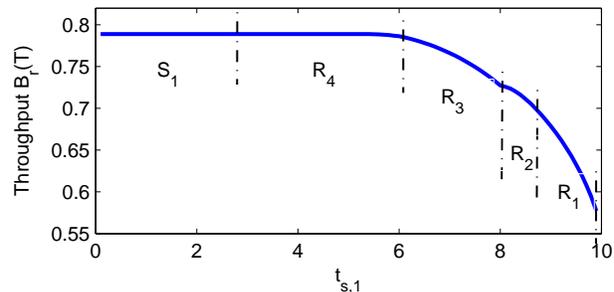}
\caption{Optimal throughput and source and relay corner point regions as a function of $t_{s,1}$ for $E_{s,0}=E_{s,1}=E_{r,0}=E_{r,1}=5$, $T=10$, $t_{r,1}=8$.}%
\label{fig 5}%
\vspace{-0.15in}
\end{figure}

\end{document}